\pdfoutput=1
\documentclass[journal,onecolumn,11pt]{IEEEtran}
\usepackage{amsmath,amssymb,amsthm}
\usepackage{graphicx}
\usepackage{tikz}
\usetikzlibrary{arrows.meta,positioning,calc}
\usepackage[hidelinks]{hyperref}

\tikzset{
  blk/.style={draw,rectangle,align=center,inner sep=4pt,minimum height=8mm},
  sig/.style={-{Latex[length=2mm]},semithick},
  lbl/.style={font=\footnotesize,inner sep=1.5pt},
}

\newtheorem{theorem}{Theorem}
\newtheorem{lemma}[theorem]{Lemma}
\newtheorem{corollary}[theorem]{Corollary}
\newtheorem{proposition}[theorem]{Proposition}
\newtheorem{remark}{Remark}

\newcommand{\st}{\le_{\mathrm{st}}}
\newcommand{\E}{\mathbb{E}}
\newcommand{\Prob}{\mathbb{P}}
\newcommand{\R}{\mathbb{R}}
\newcommand{\ones}{\mathbf{1}}

\begin{document}

\title{The Equality Cases of the Weak Simplex Conjecture\\
\Large A Response to Mulgund's Open Problem 8.1}

\author{Mengwei~Su, Kaiwen~Yang, Hao~Xu and~Chih-Lin~I,~\IEEEmembership{Life~Fellow,~IEEE}%
\thanks{M. Su and K. Yang contributed equally to this work.}%
\thanks{M. Su is with the School of Mathematical Sciences, Tongji University, Shanghai, China (e-mail: smw@tongji.edu.cn).}%
\thanks{K. Yang is with the School of Mathematical Sciences, Tongji University, Shanghai, China (e-mail: yangkaiwen@tongji.edu.cn).}%
\thanks{H. Xu is with Tongji University, Shanghai, China (e-mail: hxu@tongji.edu.cn).}%
\thanks{Chih-Lin I is with China Mobile Research Institute, Beijing, China (e-mail: icl@chinamobile.com).}}

\maketitle

\begin{abstract}
Among $n+1$ equiprobable equal-energy signals in $\R^n$ under additive white
Gaussian noise with maximum-likelihood decoding, which arrangement maximizes
the probability of correct decoding? The question is Shannon's, recorded by
Rice in 1950. Mulgund proved in 2026 that the regular-simplex value bounds
the correct-decoding probability of every signal set at every
signal-to-noise ratio, leaving open whether the simplex is the only
maximizer. This paper determines the equality cases in a form stronger than
uniqueness. A signal set other than a regular simplex falls strictly below
the bound at every positive signal-to-noise ratio. Hence a code meeting the
bound at one positive operating point is already a regular simplex, up to
vertex relabeling and an orthogonal map. In probabilistic form, among the
correlation matrices that signal sets induce, any matrix other than the
identity gives a lower-orthant probability strictly above its independent
counterpart at every finite threshold, leaving no room for a nontrivial
equality. No code of ambient dimension below $n$ attains the
bound. Under an energy budget $E$ with unrestricted blocklength the optimal
codebook is uniquely the regular simplex of circumradius $\sqrt{E}$. Every
optimal codeword therefore exhausts its allowance. Equality in the Simplex
Mean Width Conjecture likewise occurs only at the regular simplex. The
proof strengthens the first self-convolution step of Mulgund's argument
with Royen's correlation theorem. The single-parameter rigidity is
machine-checked in Lean~4.
\end{abstract}


\section{Introduction}

The Weak Simplex Conjecture concerns the oldest question of signal design for
the Gaussian channel. Rice reported in 1950 an observation of Shannon on the
placement of $n+1$ equiprobable equal-energy signals in $\R^n$ under additive
white Gaussian noise with maximum-likelihood decoding \cite{Rice1950}. The
conjectured answer is a regular simplex inscribed in the energy sphere,
claimed to maximize the probability of correct decoding at every
signal-to-noise ratio.
Massey gave the conjecture its name in his 1988 Shannon Lecture
\cite{Massey1988}.

Progress over the following seven decades was partial. Balakrishnan
established asymptotic and local optimality results \cite{Balakrishnan1961}.
Landau and Slepian proved optimality in low dimensions
\cite{LandauSlepian1966}. Later work explained why their argument does
not extend. The average-energy variant, the Strong Simplex Conjecture, is
false. Its disproof by Steiner \cite{Steiner1994} redirected attention to
the equal-energy formulation. Two announced proofs, by Goldsmith for the
mean-width form \cite{Goldsmith2021} and by Pastore for the coding form
\cite{Pastore2023}, contain gaps recorded in \cite[App.~B]{Mulgund2026}.
We follow \cite{Mulgund2026} in citing them as prior announcements without
using their asserted conclusions.

The problem acquired a second life in convex geometry. Balitskiy, Karasev
and Tsigler formulated the coding question in the language of Gaussian
maxima \cite{BKT17}. The Simplex Mean Width Conjecture asserts that among
simplices contained in the unit ball the regular simplex maximizes the mean
width. Kabluchko, Litvak and Zaporozhets recast this assertion as a
statement about expected maxima of correlated Gaussian variables and proved
an asymptotic form of it \cite{KLZ17}. The two questions are not on the same level. The
coding question compares moment-generating functions of the Gaussian maximum
at every signal-to-noise ratio, while the mean-width question compares
expectations alone. Expectation is the vanishing-ratio limit of the former
comparison. A resolution of the coding question therefore carries the
mean-width statement with it. The converse implication is not available. Sun, Hu and Lan
settled the four-variable case \cite{SHL20}. Litvak's survey records the
state of the art before these developments \cite{Litvak2018}.

Mulgund resolved the conjecture in 2026 \cite{Mulgund2026}. His
stochastic-domination theorem states that whenever an $(n+1)\times(n+1)$
correlation matrix satisfies $R-J/(n+1)\succeq 0$, the coordinate maximum of
$N(0,R)$ is stochastically dominated by the maximum of $n+1$ independent
standard Gaussian variables. The theorem yields optimality of the regular
simplex at every signal-to-noise ratio together with the mean-width
inequality. A machine-checked Lean~4 formalization accompanies it.
Its closing section poses Open Problem~8.1. The problem asks for the
equality cases of the lower-orthant comparison, of the moment-generating
comparison and of their coding and mean-width consequences. It asks in
particular whether equality at one nontrivial threshold or one positive
moment-generating parameter already forces $R=I$. Mulgund notes that his argument does not
answer this question, because the variational step produces a unique aligning tilt
but neither supplies the equality cases of the product inequality nor shows
that every optimal code is regular.

The classical form of this question asks for the best achievable
probability of correct decoding and for one arrangement that attains it.
The question left open asks for something different, the full set of
arrangements that attain it. The answer decides whether the regular
simplex is one good design among many or the only one.

This paper answers Open Problem~8.1 completely, in a form stronger than the
question requires. Stated in the language of the channel,
Theorem~\ref{thm:coding} says that a signal set other than a regular simplex
is strictly worse at every positive signal-to-noise ratio. Matching the
optimum at one positive operating point therefore identifies the code as a
regular simplex up to an orthogonal transformation. No code in ambient dimension $d<n$ attains
the regular-simplex upper bound. Under a finite energy budget with unrestricted blocklength the
optimal codebook is a regular simplex of full circumradius in which every
codeword spends its whole allowance. Theorem~\ref{thm:mgf} is the same
rigidity for the moment-generating comparison at a single positive
parameter. Both rest on Theorem~\ref{thm:threshold}, a strictly stronger
probabilistic fact. If $R\ne I$ then the lower-orthant probability exceeds
$\Phi(c)^{n+1}$ at every finite threshold. No threshold remains at which a
nontrivial equality could hide. Corollary~\ref{cor:smwc} settles the
equality case of the Simplex Mean Width Conjecture. The engine behind all of
these statements
is a strict version of the first normalized self-convolution step in
Mulgund's proof. For the truncated exponential factors produced by his
adaptive tilting, the self-convolution turns each factor into a symmetric
interval indicator. A strict symmetric rectangle inequality then yields a
strictly positive deficit. Royen's Gaussian correlation theorem
\cite{Royen2014} supplies the grouping step in its proof. The deficit
survives the remaining non-strict part of the argument.

Two consequences are worth stating for signal design. The regular simplex is
not one optimum among several. Every departure from it costs, at every
positive signal-to-noise ratio. A numerical search over constellations therefore
cannot terminate on a tie. Any tie it reports is an artifact of the
numerics. The finite-energy statement sits in the model used
for non-asymptotic analysis of short codes \cite{PPV11}, where the message
set is small and the binding constraint is the energy budget rather than the
blocklength. In that model the optimal codebook is now pinned exactly.
Moreover, the requirement that every codeword spend its full allowance is a
consequence rather than a design choice.

One ingredient may be of independent interest.
Lemma~\ref{lem:rectangle} sharpens the classical symmetric rectangle
inequality of \v{S}id\'ak and Khatri \cite{Sidak1967,Khatri1967}. For a
centered Gaussian vector with standard marginals and correlation matrix
$S\ne I$, the probability of a symmetric rectangle strictly exceeds the
product of its coordinate probabilities at every finite choice of radii.
Singular $S$ is allowed. The classical form gives only a weak
inequality. Strictness does not follow from it by approximation, since
a limit of strict inequalities need not stay strict. The proof isolates a
two-dimensional strict gap and carries it through the remaining coordinates
with Royen's theorem.

Section~II sets up the channel model and the two translations that carry the
coding question into probability and into convex geometry. Section~III states
the main results. Section~IV fixes the inputs quoted from \cite{Mulgund2026}
and from Royen's theorem. Section~V proves the two strict inequalities that carry the
strictness. Section~VI proves the theorems and the corollary. Section~VII
concludes. The appendix holds the sufficient-statistic reduction used for
the infinite-blocklength model.

\section{Channel Model and Two Translations}

This section fixes the channel and records the two changes of language that
the rest of the paper relies on. The coding question becomes an extremal
problem for the maximum of a correlated Gaussian vector. That maximum is in
turn the mean width of the simplex spanned by the signals. Both
translations are known. They are restated here because the theorems are
stated in the first language and proved in the second.

\subsection{The coding problem}

\begin{figure}[!t]
\centering
\begin{tikzpicture}
  \node[blk] (src) {Message\\source};
  \node[blk,right=11mm of src] (map) {Signal\\mapping};
  \node[draw,circle,right=13mm of map,inner sep=1.5pt] (add) {$+$};
  \node[blk,right=13mm of add] (dec) {Maximum-likelihood\\decoder};
  \coordinate[right=11mm of dec] (out);

  \draw[sig] (src) -- node[lbl,above]{$I$} (map);
  \draw[sig] (map) -- node[lbl,above]{$\lambda x_I$} (add);
  \draw[sig] (add) -- node[lbl,above]{$Y$} (dec);
  \draw[sig] (dec) -- node[lbl,above]{$\hat I$} (out);
  \node[lbl,below=8mm of add] (z) {$Z\sim N(0,I_d)$};
  \draw[sig] (z) -- (add);

  \begin{scope}[shift={($(src.south)+(6mm,-13mm)$)},scale=0.5]
    \draw (0,0) circle (1);
    \draw (90:1) -- (210:1) -- (330:1) -- cycle;
    \foreach \a in {90,210,330} \fill (\a:1) circle (0.08);
  \end{scope}
  \node[lbl,anchor=west] at ($(src.south)+(14mm,-13mm)$)
       {equal-energy signals lie on one sphere};
\end{tikzpicture}
\caption{The channel model of \eqref{eq:channel}. One of $n+1$ equiprobable
messages selects a unit-energy signal, the channel adds independent Gaussian
noise. The decoder maximizes the likelihood. The signal set for $n=2$ is
drawn on the left.}
\label{fig:system}
\end{figure}
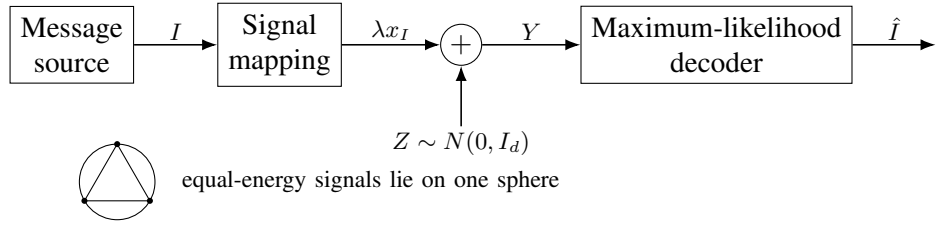

A transmitter sends one of $m=n+1$ equiprobable messages over the real AWGN
channel by mapping message $i$ to a signal $x_i\in\R^d$ of unit energy, as
in Fig.~\ref{fig:system}. The receiver observes
\begin{equation}\label{eq:channel}
Y=\lambda x_I+Z,\qquad Z\sim N(0,I_d),\ \lambda>0,
\end{equation}
where $\lambda^2$ is the signal-to-noise ratio. The receiver decodes by
maximum likelihood. Completing the square in the Gaussian likelihood turns the
decision statistic into the inner product $\langle Y,x_i\rangle$. The
average probability of correct decoding becomes
\begin{equation}\label{eq:pc-identity}
P_c(x_1,\dots,x_m;\lambda)=\frac{e^{-\lambda^2/2}}{m}\,
\E\exp\Bigl\{\lambda\max_{1\le i\le m}\xi_i\Bigr\},\qquad
\xi\sim N(0,G),\ G_{ij}=\langle x_i,x_j\rangle .
\end{equation}
The derivation runs inside the span of the signals and does not require $G$
invertible. Identity \eqref{eq:pc-identity} holds for every measurable
likelihood-maximizing decoder \cite[Cor.~2.6]{Mulgund2026}.

\subsection{From coding to Gaussian maxima}

Identity \eqref{eq:pc-identity} says that the code enters the error
probability only through the Gram matrix $G$ and that the design problem is
an extremal problem for the moment-generating function of the maximum of a
correlated Gaussian vector. Each signal set becomes a correlation matrix.
Each signal-to-noise ratio becomes a moment-generating parameter. The
regular simplex becomes the Gram matrix $G_\Delta$ of \eqref{eq:gdelta}
below. A normalization that adds one common Gaussian coordinate to every
signal, the identity \eqref{eq:mgf-transfer} below, further maps $G_\Delta$ to the identity
matrix. The question of whether the simplex is optimal thus turns into the
question of whether correlated Gaussian coordinates can beat independent
ones at their own maximum.

\subsection{From Gaussian maxima to mean width}

The same quantity is a convex-geometric functional. For a simplex $K$ with
vertices $y_1,\dots,y_m$ the support function is
$h_K(g)=\max_i\langle g,y_i\rangle$. Integrating in polar coordinates
gives $\E h_K(g)=\E\|g\|\,w(K)/2$ for $g\sim N(0,I_n)$, where $w$ denotes
mean width. The expected Gaussian maximum is therefore the mean width up to
a dimensional constant. This equivalence is due to Kabluchko,
Litvak and Zaporozhets \cite{KLZ17}. In this dictionary the low
signal-to-noise limit of \eqref{eq:pc-identity} is the mean-width
functional. Hence the Simplex Mean Width Conjecture is the $\lambda\to 0$ face
of the coding problem. A statement proved for every $\lambda>0$
automatically covers it.

\begin{figure}[!t]
\centering
\begin{tikzpicture}[x=1mm,y=1mm,every node/.style={font=\scriptsize}]
  \foreach \x in {0,48,96} \draw (\x,0) rectangle (\x+32,-36);
  \foreach \x/\t in {0/Communication, 48/Probability, 96/Geometry}
    \node[anchor=north] at (\x+16,-1) {\textsc{\t}};
  \node[anchor=north west,align=left] at (2,-6)
    {AWGN channel\\ $n+1$ equal-energy\\ signals\\ ML decoding};
  \node[anchor=north west,align=left] at (50,-6)
    {Gram matrix $G$\\ correlated Gaussian\\ coordinates\\ maximum of them};
  \node[anchor=north west,align=left] at (98,-6)
    {simplex inscribed\\ in the ball\\ mean width};
  \begin{scope}[shift={(16,-28)}]
    \foreach \a in {90,210,330} \fill (\a:3) circle (0.45);
    \foreach \p in {(1,1),(-1.6,1.4),(1.7,-1.5),(-0.9,-1.7)}
      \fill[gray] \p circle (0.25);
  \end{scope}
  \begin{scope}[shift={(64,-28)}]
    \foreach \i in {0,1,2,3} \fill (\i*3-4.5,0) circle (0.45);
    \draw[semithick] (-4.5,0) to[bend left=50] (-1.5,0);
    \draw[semithick] (-1.5,0) to[bend left=50] (4.5,0);
    \draw (4.5,0) circle (1.3);
  \end{scope}
  \begin{scope}[shift={(112,-28)}]
    \draw (0,0) circle (4.5);
    \draw (90:4.5) -- (210:4.5) -- (330:4.5) -- cycle;
    \foreach \a in {90,210,330} \fill (\a:4.5) circle (0.45);
  \end{scope}
  \draw[sig] (32,-28) -- node[above,align=center]{same Gram\\ matrix} (48,-28);
  \draw[sig] (80,-28) -- node[above,align=center]{vanishing\\ SNR} (96,-28);
  \draw (0,-41) rectangle (128,-49);
  \node[font=\footnotesize] at (64,-45)
    {The regular simplex is the strictly unique optimum in all three};
\end{tikzpicture}
\caption{The same extremal problem in three languages. A signal set with
Gram matrix $G$ fixes the correlation structure of a Gaussian vector, whose
maximum controls both the probability of correct decoding at a given
signal-to-noise ratio and, in the vanishing signal-to-noise limit, the mean
width of the simplex spanned by the signals. This paper proves that the
regular simplex is the strictly unique optimum in every column.}
\label{fig:dictionary}
\end{figure}
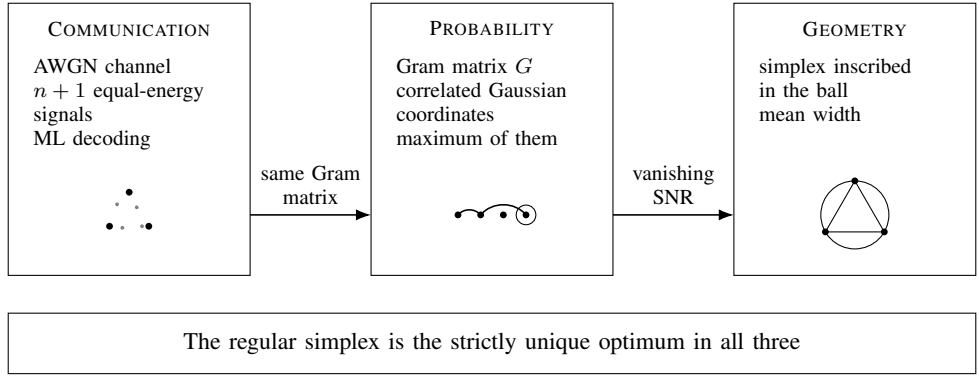

The three columns of Fig.~\ref{fig:dictionary} are the same optimization
seen from communication, probability and geometry. The results below are
stated in the first language and proved in the second. In the third they
follow as corollaries.

\section{Main Results}

Fix $n\ge 1$ and work with $n+1$ coordinates. Let $\phi$, $\Phi$ denote the
standard normal density and distribution function, $\ones=(1,\dots,1)^{\mathsf T}$,
$J=\ones\ones^{\mathsf T}$, $\alpha=n/(n+1)$, and let
\begin{equation}\label{eq:gdelta}
G_\Delta=\frac{n+1}{n}\Bigl(I-\frac{1}{n+1}J\Bigr)
\end{equation}
be the Gram matrix of the centered regular $n$-simplex inscribed in the unit
sphere.

\begin{theorem}[Single-threshold strictness]\label{thm:threshold}
Let $R$ be an $(n+1)\times(n+1)$ correlation matrix with
$R-J/(n+1)\succeq 0$, possibly singular, and let $X\sim N(0,R)$. If
$R\ne I$, then for every finite $c\in\R$,
\[
\Prob\{X\le c\ones\}>\Phi(c)^{n+1}.
\]
Consequently $\Prob\{X\le c\ones\}=\Phi(c)^{n+1}$ at some finite $c$ if and
only if $R=I$.
\end{theorem}

\begin{theorem}[Single-parameter rigidity of the MGF comparison]\label{thm:mgf}
Let $G$ be any $(n+1)\times(n+1)$ correlation matrix, $\xi\sim N(0,G)$ and
$\xi^\Delta\sim N(0,G_\Delta)$. For every $\lambda>0$,
\[
\E e^{\lambda\max_i\xi_i}\le\E e^{\lambda\max_i\xi^\Delta_i},
\]
and equality holds at some, equivalently every, $\lambda>0$ if and only if
$G=G_\Delta$. If $G\ne G_\Delta$ the inequality is strict at every
$\lambda>0$.
\end{theorem}

\begin{theorem}[Coding consequences]\label{thm:coding}
Consider $n+1$ equiprobable messages and unit signals $x_1,\dots,x_{n+1}\in\R^d$
observed as $Y=\lambda x_I+Z$, $Z\sim N(0,I_d)$, $\lambda>0$.
\begin{enumerate}
\item[(a)] The probability of correct Bayes-optimal decoding attains the
regular-simplex value at some $\lambda>0$ if and only if the code is a
centered regular $n$-simplex inscribed in the unit sphere, up to vertex
relabeling and orthogonal maps of the ambient space. If $d<n$ the
regular-simplex bound is strict for every code and every $\lambda>0$.
\item[(b)] In the deterministic no-feedback infinite-sequence AWGN model
with noise variance $N_0/2>0$, $n+1$ equiprobable messages, unrestricted
channel uses and per-codeword energy at most $E>0$, the optimal probability
of correct decoding equals $S_{\mathrm{reg}}(\sqrt{E})$, the regular-simplex
value at circumradius $\sqrt{E}$, and a codebook attains it if and only if
$\sum_i c_i=0$, $\|c_i\|^2=E$ and $\langle c_i,c_j\rangle=-E/n$ for
$i\ne j$. In particular every optimal codeword uses full energy.
\end{enumerate}
\end{theorem}

\begin{corollary}[Equality case of the Simplex Mean Width Conjecture]\label{cor:smwc}
Let $K\subset B_2^n$ be an $n$-dimensional simplex and $\Delta_n$ a regular
simplex inscribed in $S^{n-1}$ with centroid at the origin. Then
$w(K)\le w(\Delta_n)$, and $w(K)=w(\Delta_n)$ if and only if $K=U\Delta_n$
for some $U\in O(n)$. The inequality itself is \cite[Cor.~2.5]{Mulgund2026}.
Only the equality case is proved here.
\end{corollary}

\subsection{A worked example}

One family makes the statements above concrete and is computable in closed
form. Call a signal set equicorrelated when all pairwise inner products
share one value $\rho$. The family is feasible for $-1/n\le\rho<1$ and
contains the regular simplex exactly at $\rho=-1/n$. Sending $x_1$ and comparing scores,
the $n$ differences are
$\langle Y,x_1-x_j\rangle=\lambda(1-\rho)+\langle Z,x_1-x_j\rangle$. The
Gaussian part has covariance $(1-\rho)(I+J)$. It can be written as
$\sqrt{1-\rho}\,(V+U_j)$ with $V,U_1,\dots,U_n$ independent standard normal.
Dividing by $\sqrt{1-\rho}$ leaves
\begin{equation}\label{eq:worked}
P_c=\E\,\Phi\bigl(W+\lambda\sqrt{1-\rho}\bigr)^{n},\qquad W\sim N(0,1).
\end{equation}
Since $\Phi$ is increasing, \eqref{eq:worked} is strictly increasing in
$\lambda\sqrt{1-\rho}$. The factor $\sqrt{1-\rho}$ is largest at the
smallest feasible $\rho$. Inside this family the regular simplex therefore wins
strictly at every signal-to-noise ratio. This proves
Theorem~\ref{thm:coding}(a) for equicorrelated codes in three lines.
Figs.~\ref{fig:pcsnr} and~\ref{fig:deficit} plot \eqref{eq:worked} and its
probabilistic counterpart. The general case is not reachable this way, since
a general Gram matrix admits no such one-dimensional reduction. Supplying
that is the task of the rest of the paper. Neither figure is part of a proof.
Both are computed by quadrature rather than simulation.

\begin{figure}[t]
\centering
\includegraphics[width=\linewidth]{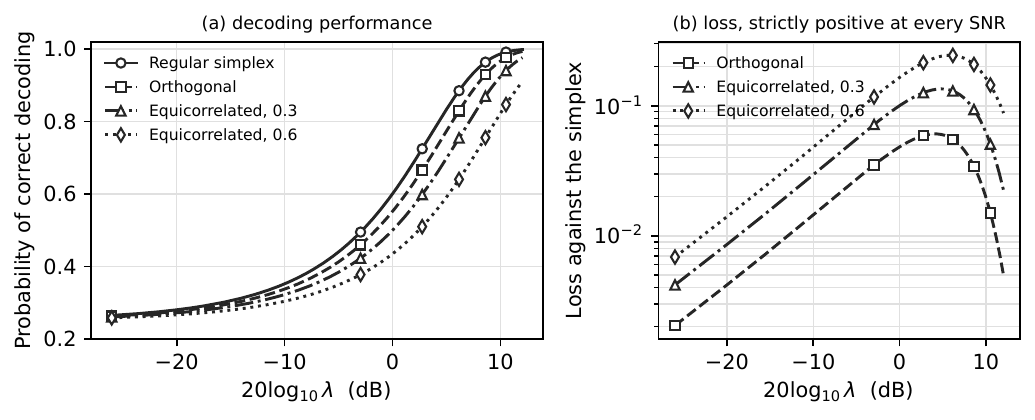}
\caption{Theorem~\ref{thm:coding}(a) for $n=3$, that is four signals. The
comparison family is equicorrelated, meaning that all pairwise inner
products equal a common value $\rho$, with the regular simplex at
$\rho=-1/3$ and mutually orthogonal signals at $\rho=0$. The ambient
dimension does not enter \eqref{eq:worked}; the family with $\rho\ge 0$
requires four dimensions, the simplex only three. Panel (a) shows the
probability of correct decoding, panel (b) the loss against the simplex on a
logarithmic scale. The loss stays strictly positive across the whole range,
which is what forbids a non-simplex code from tying the optimum at any
single operating point. Both panels are computed by Gauss--Hermite
quadrature from the closed form \eqref{eq:worked}. No simulation is involved.}
\label{fig:pcsnr}
\end{figure}

\begin{figure}[t]
\centering
\includegraphics[width=\linewidth]{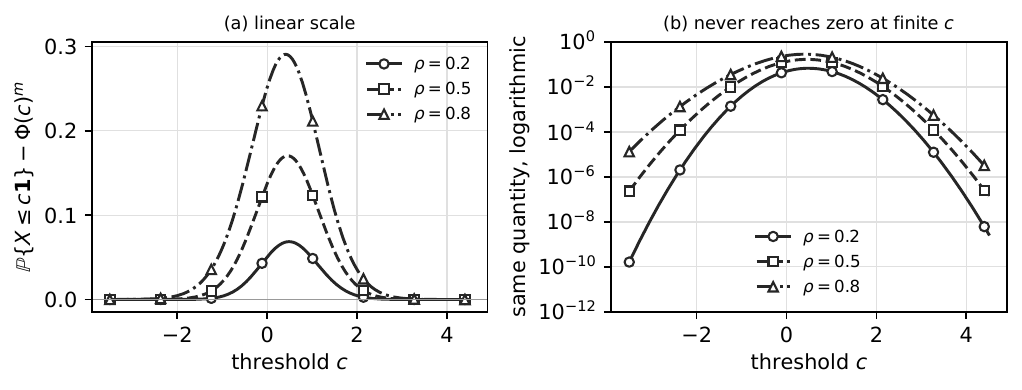}
\caption{Theorem~\ref{thm:threshold} for $m=4$ equicorrelated coordinates
with common correlation $\rho$. The plotted quantity is the amount by which
the correlated lower-orthant probability exceeds its independent
counterpart. Panel (b) repeats panel (a) on a logarithmic scale and shows
that the excess never vanishes at a finite threshold, decaying to zero only
in the two formal limits $c\to\pm\infty$ excluded by the theorem. Values are
computed by quadrature from
$\E\,\Phi((c-\sqrt{\rho}W)/\sqrt{1-\rho})^{m}$.}
\label{fig:deficit}
\end{figure}

\begin{remark}[Machine-checked companion]\label{rem:lean}
A Lean~4 companion, built on the formalization accompanying
\cite{Mulgund2026} whose axiom audit we reproduced independently, verifies
the single-parameter moment-generating rigidity along an independent route
with two steps. Stochastic domination combined with equality of one
strictly increasing moment forces equality of laws. A correlated Gaussian
vector whose maximum carries the independent distribution at all
thresholds is independent. Every public declaration depends only on the three standard
Lean axioms. The companion is at
\url{https://github.com/Tongji708A/simplex-uniqueness-lean}, commit
\texttt{6a4f340}, is pinned to
Lean 4.31.0 and mathlib v4.31.0 with the dependency fixed at an exact
revision, builds with a single \texttt{lake build}. Its axiom audit is
a source file whose expected output is recorded in the repository. The
single-threshold statement of Theorem~\ref{thm:threshold} rests
additionally on Royen's theorem. No formalization of that theorem exists to
date. Its grouping step reduces to the slab-against-convex
Khatri--\v{S}id\'ak case, a plausible target for future formalization.
\end{remark}

\section{Quoted Inputs}

The following results are quoted from \cite{Mulgund2026} and restated only
to fix notation, without repeating their proofs.

\begin{proposition}[{\cite[Thm.~2.1]{Mulgund2026}}]\label{prop:dom}
If $R$ satisfies the hypotheses of Theorem~\ref{thm:threshold} and
$X\sim N(0,R)$, then $\Prob\{X\le c\ones\}\ge\Phi(c)^{n+1}$ for every
$c\in\R$; equivalently $\max_iX_i\st\max_iZ_i$ for independent standard
Gaussian $Z_i$.
\end{proposition}

\begin{proposition}[{\cite[Thm.~4.1]{Mulgund2026}}]\label{prop:product}
Let $S$ be any $(n+1)\times(n+1)$ correlation matrix and $Y\sim N(0,S)$. If
$f_i\colon\R\to[0,\infty)$ are bounded Borel log-concave functions with
$0<\int f_i\,d\gamma_1<\infty$ and $\int zf_i(z)\,d\gamma_1(z)=0$, then
\begin{equation}\label{eq:nonstrict-product}
\E\prod_{i=1}^{n+1}f_i(Y_i)\ge\prod_{i=1}^{n+1}\int f_i\,d\gamma_1 .
\end{equation}
Singular $S$ is allowed.
\end{proposition}

Define the lower-tail inverse Mills ratio and the truncated-normal map
\[
r(s)=\frac{\phi(s)}{\Phi(s)},\qquad H(s)=s+r(s).
\]
$H$ is a strictly increasing bijection from $\R$ onto $(0,\infty)$. For
$a>0$, $b\in\R$, $s=b-a$,
\begin{equation}\label{eq:tilt-integrals}
\int e^{az}\ones_{(-\infty,b]}(z)\,d\gamma_1(z)=e^{a^2/2}\Phi(s),\qquad
\int ze^{az}\ones_{(-\infty,b]}(z)\,d\gamma_1(z)=e^{a^2/2}\bigl(a\Phi(s)-\phi(s)\bigr),
\end{equation}
and the second integral vanishes precisely when $a=r(s)$, in which case
$b=H(s)$ \cite[Lem.~3.1]{Mulgund2026}.

\begin{proposition}[{\cite[Prop.~5.2]{Mulgund2026}}]\label{prop:tilt}
Under the hypotheses of Theorem~\ref{thm:threshold} and for each finite $c$
there exist finite $s_1,\dots,s_{n+1}$ with $a_i=r(s_i)>0$ such that
\begin{equation}\label{eq:alignment}
s+a-Ra=c\ones,\qquad
\sum_i\log\Phi(s_i)+\tfrac12 a^{\mathsf T}(I-R)a\ge(n+1)\log\Phi(c).
\end{equation}
\end{proposition}

\begin{proposition}[{\cite[Lem.~6.1]{Mulgund2026}}]\label{prop:com}
If $X\sim N(0,R)$ with $R\succeq 0$, then for every $a\in\R^{n+1}$ and
nonnegative Borel $F$,
$\E[e^{a^{\mathsf T}X-a^{\mathsf T}Ra/2}F(X)]=\E F(X+Ra)$. The identity does
not involve $R^{-1}$ and holds for singular $R$.
\end{proposition}

The one external input beyond \cite{Mulgund2026} is Royen's Gaussian
correlation theorem \cite{Royen2014}. For every centered Gaussian measure
$\gamma$ on $\R^k$ and all symmetric convex sets $A,B$,
$\gamma(A\cap B)\ge\gamma(A)\gamma(B)$. The coordinatewise rectangle case is
the earlier \v{S}id\'ak--Khatri inequality \cite{Sidak1967,Khatri1967}. The
grouping below pairs a two-dimensional block with symmetric slabs, a
slab-against-convex case already contained in the Khatri--\v{S}id\'ak lemma
and covered a fortiori by Royen's theorem, quoted here as a single
self-contained reference.

\section{Two Strict Inequalities}

Theorem~\ref{thm:threshold} rests on two inequalities proved here.
The first sharpens a classical Gaussian rectangle bound. The second applies
it to the truncated exponential factors that Mulgund's adaptive tilting
produces.

\subsection{A strict symmetric rectangle inequality}

The first of the two inequalities behind Theorem~\ref{thm:threshold} is a
strict form of the \v{S}id\'ak--Khatri rectangle inequality, stated here in
the generality in which it is used and in which it may serve elsewhere.

\begin{lemma}\label{lem:rectangle}
Let $V=(V_1,\dots,V_{n+1})$ be a centered Gaussian vector with $N(0,1)$
marginals and correlation matrix $S\ne I$, possibly singular. Then for all
finite $r_1,\dots,r_{n+1}\in(0,\infty)$,
\begin{equation}\label{eq:strict-rectangle}
\Prob\{|V_i|\le r_i,\ 1\le i\le n+1\}>\prod_{i=1}^{n+1}\Prob\{|Z|\le r_i\}.
\end{equation}
\end{lemma}

\begin{proof}
Since $S\ne I$ with unit diagonal there are $i\ne j$ with
$\rho=S_{ij}\ne 0$. We first prove two-dimensional strictness. Let
$(Y_1,Y_2)$ be standard bivariate Gaussian with correlation $\rho$ and fix
$r,s>0$. Because the events involve absolute values only, replacing $Y_2$ by
$-Y_2$ if necessary reduces to $\rho_0=|\rho|$. Suppose first
$0<\rho_0<1$ and write $Y_2=\rho_0Y_1+\sigma Z$ with
$\sigma=\sqrt{1-\rho_0^2}$ and $Z$ independent of $Y_1$. For $w\ge 0$ put
\[
q(w)=\Prob\{|\rho_0w+\sigma Z|\le s\}
=\Phi\Bigl(\frac{s-\rho_0w}{\sigma}\Bigr)+\Phi\Bigl(\frac{s+\rho_0w}{\sigma}\Bigr)-1 .
\]
Symmetry of $[-s,s]$ makes the conditional probability given $Y_1=w$ and
given $Y_1=-w$ equal. With $W=|Y_1|$,
\[
\Prob\{|Y_1|\le r,\,|Y_2|\le s\}=\E\bigl[\ones_{\{W\le r\}}q(W)\bigr],\qquad
\E q(W)=\Prob\{|Y_2|\le s\}.
\]
For $w>0$,
\[
q'(w)=\frac{\rho_0}{\sigma}\Bigl[\phi\Bigl(\frac{s+\rho_0w}{\sigma}\Bigr)
-\phi\Bigl(\frac{s-\rho_0w}{\sigma}\Bigr)\Bigr]<0,
\]
because $(s+\rho_0w)^2-(s-\rho_0w)^2=4s\rho_0w>0$ and $\phi$ is strictly
decreasing in the square of its argument. Thus $q$ is strictly decreasing on
$(0,\infty)$. With $W'$ an independent copy of $W$,
\[
2\,\mathrm{Cov}\bigl(q(W),\ones_{\{W\le r\}}\bigr)
=\E\bigl[(q(W)-q(W'))(\ones_{\{W\le r\}}-\ones_{\{W'\le r\}})\bigr],
\]
whose integrand is everywhere nonnegative and strictly positive on the
positive-probability event $\{W<r<W'\}$. Hence
\begin{equation}\label{eq:2d-strict}
\Prob\{|Y_1|\le r,\,|Y_2|\le s\}>\Prob\{|Y_1|\le r\}\Prob\{|Y_2|\le s\}.
\end{equation}
If $|\rho|=1$ then $Y_2=\pm Y_1$ almost surely and, with
$p(t)=\Prob\{|Z|\le t\}$, the left side equals $p(\min(r,s))$. This minimum
exceeds $p(r)p(s)$ because both factors lie in $(0,1)$. Hence \eqref{eq:2d-strict}
holds for all $\rho\ne 0$.

In $n+1$ dimensions set $A=\{|V_i|\le r_i,\,|V_j|\le r_j\}$ and
$B_k=\{|V_k|\le r_k\}$ for $k\ne i,j$. These sets and their finite
intersections are symmetric and convex. Recursive application of Royen's
theorem gives
\[
\Prob\Bigl(A\cap\bigcap_{k\ne i,j}B_k\Bigr)\ge
\Prob(A)\prod_{k\ne i,j}\Prob(B_k),
\]
and all $\Prob(B_k)$ are strictly positive. Combining this with
\eqref{eq:2d-strict} proves \eqref{eq:strict-rectangle} for
nonsingular laws. If $S$ is singular, write $V=TY$ with
$Y\sim N(0,I_{\mathrm{rank}\,S})$ and $TT^{\mathsf T}=S$. Replace each
set by its preimage under $T$. Preimages of symmetric convex sets under a
linear map are symmetric convex. Royen's theorem then applies directly in the
nondegenerate space of $Y$. Bivariate marginals with $|\rho|<1$ and
$|\rho|=1$ are both covered above. The singular case therefore introduces no gap.
A positive-definite approximation is deliberately avoided, since limits of
strict inequalities need not stay strict.
\end{proof}

\subsection{A strict product inequality for truncated exponentials}

\begin{lemma}\label{lem:strict-product}
Let $S\ne I$ be any $(n+1)\times(n+1)$ correlation matrix and
$Y\sim N(0,S)$. For each $i$ let $a_i>0$, $b_i>0$, set
$f_i(z)=e^{a_iz}\ones_{(-\infty,b_i]}(z)$, and assume
$\int zf_i(z)\,d\gamma_1(z)=0$. Then
\begin{equation}\label{eq:strict-product}
\E\prod_{i=1}^{n+1}f_i(Y_i)>\prod_{i=1}^{n+1}\int f_i\,d\gamma_1 .
\end{equation}
Singular $S$ is allowed.
\end{lemma}

\begin{proof}
Each $f_i$ is bounded by $e^{a_ib_i}$, Borel, log-concave, strictly positive
on $(-\infty,b_i]$. The mass $\mu_i=\int f_i\,d\gamma_1$ lies in $(0,\infty)$. Normalize
$h_i=f_i/\mu_i$. With the normalized self-convolution
\[
(Dh)(u)=\int_\R h\Bigl(\frac{u+v}{\sqrt2}\Bigr)h\Bigl(\frac{u-v}{\sqrt2}\Bigr)\,d\gamma_1(v),
\]
set $Q_0=\E_S\prod_ih_i(Y_i)$ and $Q_1=\E_S\prod_i(Dh_i)(Y_i)$. Take
independent $Y,Y'\sim N(0,S)$ and put $U=(Y+Y')/\sqrt2$ and
$V=(Y-Y')/\sqrt2$. Then $U$ and $V$ are independent, each $N(0,S)$. They satisfy
\begin{equation}\label{eq:q0-squared}
Q_0^2=\E_U\E_V\prod_ig_{i,U_i}(V_i),\qquad
g_{i,u}(v)=h_i\Bigl(\frac{u+v}{\sqrt2}\Bigr)h_i\Bigl(\frac{u-v}{\sqrt2}\Bigr).
\end{equation}
The exponential part cancels in $v$, leaving for $\sqrt2\,b_i-u>0$
\begin{equation}\label{eq:interval-indicator}
g_{i,u}(v)=\mu_i^{-2}e^{\sqrt2\,a_iu}\,\ones_{\{|v|\le\sqrt2\,b_i-u\}},
\end{equation}
a symmetric interval indicator up to a positive factor. The set
$E=\{u_i<\sqrt2\,b_i,\ 1\le i\le n+1\}$ has positive $U$-probability,
because all $b_i>0$ and $U$ concentrates near the origin with positive
probability. On $E$, Lemma~\ref{lem:rectangle} applied to $V\sim N(0,S)$
with radii $\sqrt2\,b_i-u_i>0$ gives
\[
\E_V\prod_ig_{i,u_i}(V_i)>\prod_i\int g_{i,u_i}\,d\gamma_1=\prod_i(Dh_i)(u_i)
\qquad(u\in E).
\]
Off $E$ at least one radius is nonpositive. A negative radius kills the
corresponding factor identically. A zero radius gives zero on both sides
because $V_i$ has no atoms. Writing
$\Delta(u)=\E_V\prod_ig_{i,u_i}(V_i)-\prod_i(Dh_i)(u_i)$, the function
$\Delta$ is measurable, bounded by $2\prod_i\|h_i\|_\infty^2$, vanishes on
$E^c$ and is strictly positive on $E$. Therefore
\begin{equation}\label{eq:deficit}
Q_0^2-Q_1=\E\,\Delta(U)>0 .
\end{equation}
Each $Dh_i$ is bounded, log-concave by Pr\'ekopa's theorem
\cite{Prekopa1973}, has unit Gaussian integral and zero Gaussian first
moment. Proposition~\ref{prop:product} therefore applies to the $Dh_i$ and gives
$Q_1\ge 1$. With \eqref{eq:deficit}, $Q_0^2>1$, hence $Q_0>1$.
Unnormalizing yields \eqref{eq:strict-product}. The argument is not
circular, because Proposition~\ref{prop:product} is Mulgund's independently
proved non-strict theorem and strictness is injected only in the first
self-convolution step.
\end{proof}

\section{Proofs of the Main Results}

The three theorems and the corollary are proved in turn. Each proof
assembles the two inequalities of the previous section with the inputs
quoted in Section~IV.

\subsection{Theorem~\ref{thm:threshold}}

Fix a finite $c$. Proposition~\ref{prop:tilt} supplies finite $s_i$ with
$a_i=r(s_i)>0$ and $b_i=H(s_i)>0$, the positivity of $b_i$ coming from the
range $(0,\infty)$ of $H$. Define
$f_i(z)=e^{a_iz}\ones_{(-\infty,b_i]}(z)$. Equation \eqref{eq:tilt-integrals}
then gives $\int f_i\,d\gamma_1=e^{a_i^2/2}\Phi(s_i)$ and
$\int zf_i\,d\gamma_1=0$. If $R\ne I$, Lemma~\ref{lem:strict-product} gives
\[
\E\bigl[e^{a^{\mathsf T}X}\ones_{\{X\le b\}}\bigr]
>e^{\|a\|^2/2}\prod_i\Phi(s_i).
\]
Applying Proposition~\ref{prop:com} with $F=\ones_{\{x\le b\}}$ and using
$b-Ra=s+a-Ra=c\ones$ from \eqref{eq:alignment},
\[
\Prob\{X\le c\ones\}
=e^{-a^{\mathsf T}Ra/2}\,\E\bigl[e^{a^{\mathsf T}X}\ones_{\{X\le b\}}\bigr]
>\exp\Bigl\{\sum_i\log\Phi(s_i)+\tfrac12a^{\mathsf T}(I-R)a\Bigr\}
\ge\Phi(c)^{n+1},
\]
the last step by \eqref{eq:alignment}. The change-of-measure factor is
strictly positive. Strictness survives. If $R=I$ the coordinates are
independent and equality holds at every $c$. \hfill$\blacksquare$

\begin{remark}\label{rem:variational}
The variational inequality in \eqref{eq:alignment} alone cannot prove the
theorem. Its equality case is $R\ones=\ones$. For $n\ge 3$ the condition
$R\ones=\ones$ does not force $R=I$. The exclusion of all remaining
matrices is carried entirely by Lemma~\ref{lem:strict-product}.
\end{remark}

\subsection{Theorem~\ref{thm:mgf}}

Let $M_X=\max_iX_i$ and $M_Z=\max_iZ_i$ with $Z\sim N(0,I)$. Both
exponential moments are finite since $e^{\mu M}\le\sum_ie^{\mu X_i}$. If
$R\ne I$, Theorem~\ref{thm:threshold} gives
$F_X(u)>F_Z(u)$ for every $u$, where $F$ denotes the distribution function
of the corresponding maximum. The tail representation gives
\[
\E e^{\mu M_Z}-\E e^{\mu M_X}
=\lim_{\varepsilon\downarrow 0}\int_\varepsilon^\infty
\Bigl(F_X\Bigl(\frac{\log t}{\mu}\Bigr)-F_Z\Bigl(\frac{\log t}{\mu}\Bigr)\Bigr)dt>0
\]
for every $\mu>0$, the strict sign visible on any compact
$[a,b]\subset(0,\infty)$ where the integrand is everywhere positive.

For a general correlation matrix $G$ put $R=\alpha G+J/(n+1)$, take
$B\sim N(0,1/n)$ independent of $\xi\sim N(0,G)$ and set
$X_i=\sqrt{\alpha}(\xi_i+B)$, so that $X\sim N(0,R)$ and, for
$\mu=\lambda/\sqrt{\alpha}$,
\begin{equation}\label{eq:mgf-transfer}
\E e^{\mu M_X}=e^{\lambda^2/(2n)}\,\E e^{\lambda\max_i\xi_i}.
\end{equation}
The same construction for $G_\Delta$ produces exactly $R=I$. Since
$\alpha>0$, $R=I$ if and only if $G=G_\Delta$. Combining
\eqref{eq:mgf-transfer} with the strict comparison above proves the theorem.
\hfill$\blacksquare$

\subsection{Theorem~\ref{thm:coding}}\label{ssec:coding}

The two parts of the theorem are proved separately. The equal-energy case
follows from Theorem~\ref{thm:mgf} once the Gram matrix is identified.
The finite-energy case reduces to it after an exact sufficient statistic
removes the infinite blocklength.

\paragraph{Equal energy}

For unit signals the Bayes-optimal probability of correct decoding is given
by \eqref{eq:pc-identity}, whose prefactor does not depend on the code. By
Theorem~\ref{thm:mgf}
equality at any $\lambda>0$ holds if and only if $G=G_\Delta$. The
condition $G=G_\Delta$ forces
$\langle x_i,x_j\rangle=-1/n$ for $i\ne j$ and $\sum_ix_i=0$. Fixing
reference vertices $u_i$ with the same Gram matrix, the correspondence
$\sum_ia_iu_i\mapsto\sum_ia_ix_i$ is well defined and inner-product
preserving. After completing orthonormal bases it extends to an
orthogonal map of the ambient space. The code is therefore the centered regular simplex up
to relabeling and ambient isometry. If $d<n$ then
$\mathrm{rank}\,G\le d<n=\mathrm{rank}\,G_\Delta$. Equality is therefore
impossible and the bound is strict at every $\lambda$. Finally, on the
regular simplex the maximum-likelihood scores of two distinct messages tie
only on a finite union of hyperplanes, a null set under every output law. Any Bayes-optimal decoder agrees almost everywhere with maximum-likelihood
decoding.

\paragraph{Finite total energy}

The infinite-sequence model is reduced to finitely many dimensions by an
exact sufficient statistic.

\begin{lemma}\label{lem:sufficient}
In the model of Theorem~\ref{thm:coding}(b) with codewords
$c_i\in\ell_2$, let $D=\mathrm{span}\{c_i-c_1\}$, $r=\dim D$, and let
$u_1,\dots,u_r$ be an orthonormal basis of $D$. The limits
$T_j(Y)=\lim_N\sum_{k\le N}u_{j,k}Y_k$ exist almost surely and in $L^2$
under every message law, are jointly Gaussian with
$T(Y)\mid\{I=i\}\sim N((\langle c_i,u_j\rangle)_j,\sigma^2I_r)$, and $T$ is
sufficient; the sequence experiment and the $r$-dimensional Gaussian shift
experiment have the same optimal probability of correct decoding.
\end{lemma}

\begin{proof}
See Appendix~\ref{app:sufficient}.
\end{proof}

For the assembly, let $S_{\mathrm{reg}}(r)$ denote the optimal correct
probability of the centered regular simplex of circumradius $r$. It
evaluates to
\begin{equation}\label{eq:sreg}
S_{\mathrm{reg}}(r)=p_{n+1}\Bigl(r\sqrt{\tfrac{2(n+1)}{nN_0}}\Bigr),\qquad
p_{n+1}(a)=\E\,\Phi(W+a)^n,\ W\sim N(0,1),
\end{equation}
and $S_{\mathrm{reg}}$ is strictly increasing, by differentiation under the
integral.

Let $C=\{c_i\}$ be an arbitrary feasible codebook. Augment each codeword by
one extra coordinate, $\tilde c_i=(c_i,q_i)$ with $q_i=\sqrt{E-\|c_i\|^2}$,
so that $\|\tilde c_i\|^2=E$. The augmented codebook $\tilde C$ is feasible
because channel uses are unrestricted and a fixed coordinate rearrangement
leaves the i.i.d.\ noise law invariant. It satisfies
$S(C)\le S(\tilde C)$ because a decoder may ignore the extra coordinate.
Let $\tilde D=\mathrm{span}\{\tilde c_i-\tilde c_1\}$, write
$\tilde c_i=w+v_i$ with $w$ the common projection onto $\tilde D^\perp$ and
$v_i\in\tilde D$. All $v_i$ then share one radius $r_\ast$ with
$r_\ast^2=E-\|w\|^2$. If $r_\ast=0$ the effective means coincide and the
correct probability is $1/(n+1)$, below the regular-simplex value. For
$r_\ast>0$, normalizing $x_i=v_i/r_\ast$ and applying
\eqref{eq:pc-identity} inside $\tilde D$ with $\lambda=r_\ast/\sigma$,
Theorem~\ref{thm:mgf} gives
$S(\tilde C)\le S_{\mathrm{reg}}(r_\ast)\le S_{\mathrm{reg}}(\sqrt{E})$,
the second inequality strict for $r_\ast<\sqrt{E}$ by \eqref{eq:sreg}.
Every feasible codebook therefore satisfies
$S(C)\le S_{\mathrm{reg}}(\sqrt{E})$. The regular simplex of circumradius
$\sqrt{E}$ is feasible and attains this value. Hence
$P^*_c(E,n+1)=S_{\mathrm{reg}}(\sqrt{E})$, attained by the regular simplex.

Now let $C$ attain the optimum. Every inequality above is then an
equality. The case $r_\ast=0$ is excluded because
$1/(n+1)<S_{\mathrm{reg}}(\sqrt{E})$ for $E>0$. Strict monotonicity of
$S_{\mathrm{reg}}$ forces $r_\ast=\sqrt{E}$, hence $w=0$. Equality in the
MGF comparison forces $\mathrm{Gram}(x_1,\dots,x_{n+1})=G_\Delta$ and
$\sum_ix_i=0$ by Theorem~\ref{thm:mgf}. Then
$\sum_i\tilde c_i=(n+1)w+r_\ast\sum_ix_i=0$. The last coordinate reads
$\sum_iq_i=0$ with every $q_i\ge 0$. Hence all $q_i=0$. Every codeword uses full energy and the codebook satisfies the
stated Gram conditions. Conversely any such codebook is a centered regular
simplex of circumradius $\sqrt{E}$ and attains $S_{\mathrm{reg}}(\sqrt{E})$.
A common translation is excluded by averaging the energy constraint.
\hfill$\blacksquare$

\subsection{Corollary~\ref{cor:smwc}}

Write $K=\mathrm{conv}\{y_1,\dots,y_{n+1}\}\subset B_2^n$ and
$\sigma_i=\sqrt{1-\|y_i\|^2}$. Lifting $x_i=(y_i,\sigma_ie_i)$ to unit
vectors and conditioning on $g\sim N(0,I_n)$, the two-case Jensen argument,
$\E\max(Y,C)-\max(a_j,C)$ equals $\E(Y-C)^+>0$ when $a_j\le C$ and
$\E(C-Y)^+>0$ when $a_j>C$ for the nondegenerate normal $Y=a_j+\sigma_jW_j$,
shows that any vertex strictly inside the ball produces a strict gap in
expected Gaussian maxima. Mean-width equality therefore forces all vertices onto
the sphere. On the sphere, put $G=(\langle y_i,y_j\rangle)$ and
$X_i=\sqrt{\alpha}(\xi_i+B)$ with $B\sim N(0,1/n)$ independent of $\xi$,
as in \eqref{eq:mgf-transfer}. If
$G\ne G_\Delta$ then $R\ne I$ and Theorem~\ref{thm:threshold} gives a strict
gap between the distribution functions of the maxima at every point, hence
$\E M_X<\E M_Z$ by integrating the gap, contradicting mean-width equality
through the polar decomposition $\E h_K(g)=\E\|g\|\,w(K)/2$. Thus
$G=G_\Delta$. Gram rigidity as in Section~\ref{ssec:coding}(a) identifies
$K=U\Delta_n$. Orthogonal images conversely attain equality.
\hfill$\blacksquare$

\section{Conclusion}

The equal-energy problem for $n+1$ signals is now closed in a strong sense.
To the optimality proved in \cite{Mulgund2026}, this paper adds the
equality cases. Away from the regular simplex no equality occurs at any
finite threshold, at any positive moment-generating parameter, or under any
positive energy budget. Steiner's counterexample \cite{Steiner1994}
delimits the scope, showing that under an average-energy constraint the
regular simplex is not optimal and no such rigidity can hold.

Boundary conditions are sharp and cannot be dropped. The threshold must be
finite, the moment-generating parameter strictly positive, the energy
strictly positive, since each formal endpoint produces a trivial equality.
All arguments admit singular correlation matrices. The case $|R_{ij}|=1$
is handled inside the strict rectangle lemma.

Two directions remain. The pair-tail structure of the strict rectangle
inequality is explicit and suggests a quantitative stability estimate, a
deficit bound in terms of the distance from $G$ to $G_\Delta$. For
$m\ne n+1$ signals the optimal configuration is unknown in general, the
rank-constrained problem posed in \cite{Mulgund2026} included. The
biorthogonal conjecture for $m=2n$ remains open.

\appendix[Proof of Lemma~\ref{lem:sufficient}]
\label{app:sufficient}

If $r=0$ all codewords coincide and the empty statistic is sufficient. Assume $r\ge 1$. Under message $i$,
$\sum_{k\le N}u_{j,k}Y_k=\sum_{k\le N}u_{j,k}c_{i,k}+\sum_{k\le N}u_{j,k}Z_k$.
The signal sum converges to $\langle c_i,u_j\rangle$ since $c_i,u_j\in\ell_2$.
The noise sum is a series of independent centered variables with
$\sum_k\mathrm{Var}(u_{j,k}Z_k)=\sigma^2<\infty$. It converges almost
surely and in $L^2$. The convergence set $A_j\subset\R^{\mathbb N}$ is Borel
by the Cauchy criterion. The intersection $A=\bigcap_jA_j$ carries full
measure under every message law. Setting $T_j=0$ off $A$ produces one Borel statistic
shared by all hypotheses. Finite linear combinations of the partial noise
sums converge in $L^2$. Their characteristic functions converge to
Gaussian ones. Cram\'er--Wold identifies the limit law as
$N((\langle c_i,u_j\rangle)_j,\sigma^2I_r)$, the covariance coming from
orthonormality of the $u_j$.

For sufficiency write $d_i=c_i-c_1=\sum_j\beta_{ij}u_j$ and, on the first
$N$ coordinates,
\begin{equation}\label{eq:lr-finite}
L_{i,N}(Y)=\exp\Bigl\{\frac{1}{\sigma^2}\sum_{k\le N}d_{i,k}(Y_k-c_{1,k})
-\frac{1}{2\sigma^2}\sum_{k\le N}d_{i,k}^2\Bigr\},
\end{equation}
the likelihood ratio of message $i$ against message $1$. Under $P_1$ the
sequence $(L_{i,N})_N$ is a nonnegative mean-one martingale with
$\E_1L_{i,N}^2=\exp(\sigma^{-2}\sum_{k\le N}d_{i,k}^2)
\le\exp(\|d_i\|^2/\sigma^2)$. It therefore converges in $L^2$ and in $L^1$,
and on $A$ its almost-sure limit is
\begin{equation}\label{eq:lr-limit}
L_i(Y)=\exp\Bigl\{\frac{1}{\sigma^2}\sum_{j}\beta_{ij}T_j(Y)
-\frac{\|c_i\|^2-\|c_1\|^2}{2\sigma^2}\Bigr\},
\end{equation}
a function of $T$ alone. For every cylinder event $B$ the $L^1$ convergence
gives $P_i(B)=\E_1[\ones_BL_i]$. The monotone class theorem extends this
identity to the whole product $\sigma$-algebra. Therefore
$dP_i/dP_1=L_i$. Posterior probabilities under the uniform prior are
$\pi_i=L_i/\sum_kL_k$, functions of $T$. Any decoder $\delta$ satisfies
$\Prob\{\delta(Y)=I\}\le\E\max_i\pi_i(T)$. The bound is attained by the
measurable selection $\hat\delta(t)=\min\arg\max_i\pi_i(t)$. The
selection $\hat\delta$ depends on $T$ only. A randomized decoder is a convex combination bounded the
same way. The two experiments therefore share their optimal probability of
correct decoding.

\bibliographystyle{IEEEtran}
\bibliography{references}

\end{document}